\documentclass[conference]{IEEEtran}
\usepackage[utf8]{inputenc}
\usepackage{amsmath, amssymb, amsthm}
\usepackage{enumerate}
\usepackage{graphicx}
\usepackage{xcolor}
\usepackage{booktabs}

\usepackage{algorithm, algorithmicx, algpseudocode}
\algrenewcommand\algorithmicrequire{\textbf{Input:}}
\algrenewcommand\algorithmicensure{\textbf{Output:}}
\algnewcommand{\BlackBox}[1]{%
    \begin{flushleft}
    \hspace{-.7cm}
    \textbf{Available Functions:}
    {\raggedright #1}
    \end{flushleft}
}
\algnewcommand{\Initialize}[1]{%
    \begin{flushleft}
    \hspace{-.7cm}
    \textbf{Initialize:}
    {\raggedright #1}
    \end{flushleft}
}

\usepackage{hyperref}
\usepackage[capitalize,nameinlink]{cleveref}
\makeatletter
\newcommand{\ALG@lineautorefname}{Step}
\makeatother

\usepackage{centernot}

\usepackage{color, soul}

\usepackage{mdframed}
\mdfdefinestyle{MyFrame}{%
    linecolor=black,
    outerlinewidth=2pt,
    roundcorner=20pt,
    innertopmargin=\baselineskip,
    innerbottommargin=\baselineskip,
    innerrightmargin=10pt,
    innerleftmargin=10pt,
    backgroundcolor=gray!20!white
}

\newcommand{\CC}{{\mathcal C}}

\newcommand{\CF}{{\mathcal F}}

\newcommand{\fpkq}{\CF_{p}(k, q)}

\newcommand{\BE}{{\mathbb E}}
\newcommand{\BF}{{\mathbb F}}

\newcommand{\fq}{\BF_{q}}

\newcommand{\fqn}{\BF_{q}^{n}}
\newcommand{\fqx}{\BF_{q}[X]}

\newcommand{\dH}{d_{{\rm H}}}

\newcommand{\bigabs}[1]
{{\raisebox{-0.25\depth}{$\biggl\lvert$}}{#1}\raisebox{-0.25\depth}{$\biggr\rvert$}}

\DeclareMathOperator{\CCP}{CP}
\DeclareMathOperator{\RS}{RS}
\DeclareMathOperator{\GRS}{GRS}

\DeclareMathOperator{\Vol}{Vol}
\DeclareMathOperator*{\expec}{\mathbb{E}}
\DeclareMathOperator{\GS}{GS}
\DeclareMathOperator{\GRScover}{GRS-cover}
\DeclareMathOperator{\GRScoverbl}{GRS-cover-baseline}
\DeclareMathOperator{\GRSdecode}{GRS-decode}

\DeclareMathOperator{\polylog}{polylog}

\newtheorem{theorem}{Theorem}

\newtheorem{corollary}{Corollary}
\newtheorem{conjecture}{Conjecture}

\theoremstyle{definition}

\theoremstyle{remark}

\newcommand{\question}[1]
{
\noindent {\color{red} {\bf Question:} #1}}

\newcommand{\answer}[1]
{\noindent {\color{blue} {\bf Answer:} #1}}

\newcommand{\todo}[1]
{\noindent {\color{red} {\bf To Do:} #1}}

\newcommand{\hide}[1]
{{\iffalse #1 \fi}}
\usepackage{subcaption}
\usepackage{lipsum}

\IEEEoverridecommandlockouts
\title{Efficient Covering Using Reed--Solomon Codes}
\author{\IEEEauthorblockN{Samin Riasat and Hessam Mahdavifar} 
\IEEEauthorblockA{Department of Electrical and Computer Engineering, Northeastern University, Boston, MA 02115, USA \\ 
Email: \{
\href{mailto:riasat.s@northeastern.edu}{riasat.s}, \href{mailto:h.mahdavifar@northeastern.edu}{h.mahdavifar}\}@northeastern.edu}
\thanks{This work was supported by NSF under Grant CCF-2415440 and the Center for Ubiquitous Connectivity (CUbiC) under the JUMP 2.0 program.}

}
\date{}

\begin{document}

\maketitle

\begin{abstract}
    We propose an efficient algorithm to find a Reed--Solomon (RS) codeword at a distance within the covering radius of the code from any point in its ambient Hamming space. To the best of the authors' knowledge, this is the first attempt of its kind to solve the covering problem for RS codes. The proposed algorithm leverages off-the-shelf decoding methods for RS codes, including the Berlekamp–Welch algorithm for unique decoding and the Guruswami-Sudan algorithm for list decoding. We also present theoretical and numerical results on the capabilities of the proposed algorithm and, in particular, the average covering radius resulting from it. Our numerical results suggest that the overlapping Hamming spheres of radius close to the Guruswami--Sudan decoding radius centered at the codewords cover \textit{most} of the ambient Hamming space. 
\end{abstract}

\section{Introduction}
\label{sec:intro}

\emph{Reed--Solomon (RS) codes} are a well-known family of codes that are not only of theoretical interest but also very useful in various applications, ranging from hard-disk drives to distributed storage and computing. 
RS codes belong to a class of codes known as \emph{maximum distance separable (MDS)} codes, implying that they achieve \emph{Singleton bound}~\cite{Huffman03}. 
Although various properties of these codes have been studied extensively in the literature, RS codes continue to inspire the development of new codes to this day. One such example is the class of subspace codes known as \emph{character-polynomial (CP) codes}, recently introduced in \cite{Hessam22, Hessam21}. A one-dimensional CP code is essentially the concatenation of a subcode of \emph{generalized Reed--Solomon (GRS) code}, where a certain subset of the coefficients of the message polynomial is first set to zeros, followed by applying a character function to the codeword coordinates mapping the finite field elements to the complex unit circle. In a recent work~\cite{gooty2025precodingdesignlimitedfeedbackmiso}, CP codes have been used to construct new precoding designs for the problem of multiple-input single-output (MISO) systems with limited feedback, an important problem in wireless communications. In particular, it is established in \cite{gooty2025precodingdesignlimitedfeedbackmiso} that the covering radius of CP codes can be used to characterize the quantization error of the corresponding quantization problem in the Grassmann space. 


In general, covering can be thought of as quantization in a general metric space, such as the Hamming space, which is a useful tool in scenarios that require approximating, representing, or covering continuous spaces with discrete elements. Covering problems are mathematically appealing in their own right, and have found a wide range of technical applications in data compression, signal processing, clustering, sampling, and robust system design, especially in high-dimensional spaces, just to name a few~\cite{torquato10, toth22}
. In a basic covering problem, we have a vector space over a certain alphabet that we wish to cover with as few spheres of a given radius, called the covering radius, as possible. This means that we can approximate any point in the space by mapping it to the center of one of the spheres that cover this point. This provides a guarantee on the accuracy of the approximation in terms of the covering radius. Some of the prior work on the covering problem include characterizing upper and lower bounds on the covering radius and the complexity of computing it \cite{Cohen85}, deriving bounds on the minimal size of a code of a given covering radius~\cite{Cohen86, vanLint88}, studying the complexity of bounding the covering radius of a binary code~\cite{Cohen97}, deriving relations between the covering radius of a code and its subcodes~\cite{Brualdi98}, among others. We refer the reader to the extensive list of references following \cite[Chapter~15]{Huffman03}. 

In this work, our goal is to solve the problem of covering using GRS codes. Although the covering radius of GRS codes is known, the problem of finding an explicit codeword within the covering radius of any given vector in the ambient space of the code has received little attention. 
Moreover, GRS codes are known to achieve the redundancy bound~\cite[Corollary~11.1.3]{Huffman03}. In other words, their worst-case guarantee for the covering problem in the Hamming space is not appealing. However, as discussed earlier, GRS codes can be mapped to the complex domain as in CP codes and used to cover Grassmann space with applications in wireless communications. Furthermore, the \textit{average} covering radius of RS codes could fall far below their worst case of the redundancy bound, which provides further motivation to study efficient covering algorithms for RS codes. 

In this paper, we propose an efficient algorithm to solve the covering problem for GRS codes. That is, given an arbitrary vector in the ambient Hamming space containing the code, we demonstrate how to efficiently find a codeword at a distance within the covering radius of the code from the given vector. 
Our algorithm is based on successive puncturing of the code together with leveraging an off-the-shelf decoder for the resulting GRS codes after each puncture. We then analyze the average number of punctures needed to find a codeword within the covering radius of the code. 
We also estimate the average covering radius resulting from the algorithm by computing the average distance of the output codeword from the input vector. The performance of the algorithm is compared with the worst-case guarantee of the covering radius as well as a straightforward baseline approach. 


The remainder of the paper is organized as follows. 
In \autoref{sec:background} we provide some background on Reed--Solomon codes and the sphere covering problem. In \autoref{sec:covering} we present our covering algorithm and associated theoretical results. In \autoref{sec:simulation} simulation results supporting the theoretical results obtained in \autoref{sec:covering} are presented. Finally, we conclude the paper in \autoref{sec:conclusion}. 

\section{Background}
\label{sec:background}

\subsection{Reed--Solomon Code}

Fix $k \le n \le q$. The \emph{message space}
\begin{align*}
    \CF(k, q) 
    &:= \{f \in \fqx: \deg(f) \le k\} 
\end{align*}
consists of all polynomials of degree at most $k$ over $\fq$. The elements of $\CF(k, q)$ are called \emph{message polynomials}, whose coefficients represent message symbols. 
    
Given distinct $\alpha_{1}, \dots, \alpha_{n} \in \fq$, the encoding of $f \in \CF(k - 1, q)$ 
in the RS code of length $n$ and dimension $k$ over $\fq$ is given by
\begin{align}
    \label{eq:rs}
    \RS(f) 
    &:= (f(\alpha_{1}), \dots, f(\alpha_{n})). 
\end{align}
In addition, given not necessarily distinct 
$v_{1}, \dots, v_{n} \in \fq^{\times}$, the encoding of $f \in \CF(k - 1, q)$ in the GRS code of length $n$ and dimension $k$ over $\fq$ is given by
\begin{align}
    \label{eq:grs}
    \GRS(f) 
    &:= (v_{1} f(\alpha_{1}), \dots, v_{n} f(\alpha_{n})).
\end{align}

As mentioned in \autoref{sec:intro}, $\RS$ and $\GRS$ above are well known to be MDS, i.e., they are $[n, k, d]_{q}$ codes 
with $d := n - k + 1 \ge 1$. 
Note also that $\GRS = \RS$ when $v_{i} = 1$ for all $i \in \{1, \dots, n\}$. 

\subsection{Covering Radius}

The \emph{covering radius} of a block code $\CC \subseteq \fq^{n}$ is defined as (see, e.g., \cite{Huffman03})
\begin{align}
    \label{eq:rho}
    \rho(\CC) 
    &:= \max_{y \in \fqn} 
    \min_{c \in \CC} 
    \dH(y, c), 
\end{align}
where $\dH$ denotes Hamming distance. 
For instance, it is well known that $\rho(\GRS) = d - 1$. 
This means, in particular, that any $y \in \fq^{n}$ is at distance at most $d - 1$ from a $\GRS$ codeword. 

Generally speaking, the covering radius of a code represents the \textit{maximum error} when the code is used for \emph{quantization} of the space\hide{ as described in \autoref{sec:intro}}. This is because the covering radius characterizes the worst-case scenario of the quantization process. In practice, however, given a certain distribution over the space, the average quantization error becomes more relevant. This motivates us to define the \emph{average covering radius} as 
\begin{align}
    \label{eq:rhobar}
    \bar\rho(\CC) 
    &:= \expec_{y \in \fqn} 
    \left[\min_{c \in \CC} 
    \dH(y, c)\right]. 
\end{align}

\subsection{Sphere Packing Versus Sphere Covering}
\label{sec:packing-covering}

In \emph{sphere packing}, the goal is to fit in the space pairwise disjoint spheres centered at the codewords. This is fundamentally related to decoding, where, given a point in the space, one wishes to determine its nearest codeword. In particular, given a packing of the space where the spheres have radius $\tau < d / 2$, \hide{where $d$ is the minimum distance of the code, }one can uniquely decode any message from $\tau$ errors, and the goal is to maximize the decoding radius $\tau$. The dual problem to sphere packing is \emph{sphere covering}, where the goal is to cover the space using minimally overlapping spheres. 
This is fundamentally related to quantization, which is the process of approximating any given point in the space by its closest codeword.
More formally, given a vector $y \in \fq^{n}$ and a code $\CC \subseteq \fq^{n}$, the \emph{covering problem} is to find a codeword $c \in \CC$ such that $\dH(y, c) \le \rho(\CC)$. 
An algorithm that solves the covering problem for $\CC$ is called a \emph{covering algorithm} for $\CC$. 

\section{Covering Algorithms for GRS Codes}
\label{sec:covering}

Decoding algorithms have been the subject of study for many decades, and several decoding algorithms for GRS codes are known~\cite{Huffman03, Guruswami06}. 
In contrast, little to no work has been done on covering algorithms to the best of our knowledge. 
In particular, we are unaware of any covering algorithms for GRS codes. 
Note that covering involves solving a more intricate minimax problem, which is computationally harder due to the need to analyze the worst-case scenario over the entire space. 
And, although one might use decoding algorithms for covering, their performance is bound to be subpar as spheres of radius $\tau < d / 2$ do not cover the space well. 
Therefore, as the first algorithm of its kind, we propose \autoref{alg:GRS-cover}. 
Given an $[n, k, d]$ GRS code $\CC$ and an input vector $y$ in its ambient space, the algorithm first tries to decode $y$. If the decoder is successful, then we have found the closest codeword to $y$. If not, it repeatedly (up to $d - 1$ times) punctures the code at the last coordinate and tries to decode the corresponding substring of $y$ until the decoder is successful,
when it finally returns the re-encoding of the decoder output in the original code. 

\begin{algorithm}[!htbp]
    \caption{$\GRScover(\CC, y)$} 
    \label{alg:GRS-cover}
    \begin{algorithmic}[1]
        \Require{$[n, k, d]_{q}$ $\GRS$ code $\CC$, arbitrary vector $y \in \fq^{n}$}
        \Ensure{A codeword $c \in \CC$ with $\dH(y, c) \le d - 1$}
        \BlackBox{$\GRSdecode$: a $\GRS$ decoder}
        \State{$\CC_{0} \gets \CC$}
        \For{$i = 1, \dots, d-1$}
            \State $f \gets \GRSdecode(\CC, y)$
            \label{line:rsdecode}
            \If {$f$ is not null}
                \State{\Return $\CC_{0}(f)$} 
            \Else 
                \State{puncture $\CC$ at coordinate $n - i + 1$}
                \label{line:punc}
                \State{$y \gets y[1..n - i]$} 
                \label{line:y}
            \EndIf
        \EndFor
    \end{algorithmic}
\end{algorithm}

\autoref{thm:grs-cover} below proves the correctness of \autoref{alg:GRS-cover}. 

\begin{theorem}
    \label{thm:grs-cover}
    Given an $[n, k, d]_{q}$ GRS code $\CC$ and $y \in \fq^{n}$, $\GRScover(\CC, y)$ returns a codeword $c \in \CC$ with $\dH(y, c) \le d - 1$. 
\end{theorem}

\begin{proof}
    Puncturing an $[n, k]$ GRS code at any coordinate gives an $[n - 1, k]$ GRS code. 
    Denote by $y^{(i)}$ and $\CC_{i}$ the 
    values of $y$ and $\CC$ on lines \ref{line:y} and \ref{line:punc}, respectively, at step $i$. 
    Plainly,
    \begin{align}
        \dH(y^{(i+1)}, \CC_{i + 1}(f)) 
        &\le \dH(y^{(i)}, \CC_{i}(f)) 
        \label{eq:reduction}
    \end{align}
    for any $i$. 
    If $\GRSdecode(\CC_{i}, y^{(i)})$ 
    is successful, then 
    \begin{align*}
        \dH(y^{(i)}, \CC_{i}(f)) \le \dH(y, \CC(f)) \le d - 1 
    \end{align*}
    by repeated application of \eqref{eq:reduction}. 
    It remains to show that $\GRSdecode(\CC_{i}, y^{(i)})$ does indeed succeed for some $i < d$. 
    
    Note that 
    \begin{align}
        d_{\min}(\CC_{i + 1}) = d_{\min}(\CC_{i}) - 1
        \label{eq:dmin-reduction}
    \end{align}
    for each $i$. 
    Hence, $\GRSdecode(\CC_{i}, y^{(i)})$ will succeed in at most $d - 1$ steps, since $d_{\min}(\CC_{i}) = 1$ when $i = d - 1$, at which point any $y^{(i)}$ becomes a valid codeword.
\end{proof}

In \autoref{sec:average-covering-radius}, we compare the performance of \autoref{alg:GRS-cover} with the following straightforward algorithm that we pick as the baseline. This algorithm simply punctures the code at a fixed set of (e.g., the last) $d - 1$ coordinates, followed by decoding the punctured vector in the resulting rate-1 code (equivalent to an interpolation) and re-encoding to return a codeword in the original code within the covering radius of $d-1$.\hide{ and show that it significantly improves the average covering radius.} 



\begin{algorithm}[!htbp]
    \caption{$\GRScoverbl(\CC, y)$} 
    \label{alg:trivial-cover}
    \begin{algorithmic}[1]
        \Require{$[n, k, d]_{q}$ $\GRS$ code $\CC$, arbitrary vector $y \in \fq^{n}$}
        \Ensure{A codeword $c \in \CC$ with $\dH(y, c) \le d - 1$}
        \BlackBox{$\GRSdecode$: a $\GRS$ decoder}
        \State{$\CC_{0} \gets \CC$}
        \State{puncture $\CC$ at coordinates $\{n - i + 1: 1 \le i \le d - 1\}$}
        \State{$y \gets y[1..n - d + 1]$} 
        \State $f \gets \GRSdecode(\CC, y)$
        \State{\Return $\CC_{0}(f)$} 
    \end{algorithmic}
\end{algorithm}

\subsection{Improvement Using List Decoding}

Consider \autoref{alg:GRS-cover} where $\GRSdecode$ is a list decoder, which, for us, will mostly be the Guruswami--Sudan list decoder (GS) \cite{Guruswami06}. 
There are several advantages of using a list decoder over a unique decoder. 
First, it was shown by McEliece \cite{McEliece03, McEliece03-1} that GS almost always returns a list of size $1$, whereby implying that it is essentially a unique decoder. 
Next, the decoding radius of GS is $\tau_{\GS} := n - 1 - \lfloor \sqrt{(k - 1) n} \rfloor$. 
This implies, in particular, that the modified version of \autoref{alg:GRS-cover} will return a $\GRS$ codeword $c$ within distance $\tau_{\GS}$ of $y$, i.e., with $\dH(y, c) \le \tau_{\GS}$, if such a codeword exists~\cite{McEliece03-1}. 
This improves the average number of punctures, whence also the average covering radius. 
\hide{
\todo{analyze the number of punctures analytically for unique decoder} 

\todo{prove the same for list decoder or state conjecture} 

\hide{Result about how list decoding gives better results than unique decoding?
(Check if this is accurate, or say it is an approximation.) }

demonstrated below. Let $P(\CC, y)$ denote the number of punctures needed for $\GRScover(\CC, y)$ to succeed.
\todo{don't write theorem, just text} 

\begin{theorem}
    For an $[n, k, d]_{q}$ GRS code $\CC$ and a uniformly random $y \in \fq^{n}$, the expected number of punctures needed for $\GRScover(\CC, y)$ to succeed is 
    \begin{align*}
        \BE[P(\CC, y)] 
        &= 
        \sum_{i = 0}^{d - 1} 
        i \binom{n - i}{\tau_{i}} 
        \left(1 - \frac{1}{q}\right)^{\tau_{i} + i} 
        \left(\frac{1}{q}\right)^{n - i - \tau_{i}} \\ 
        &= \frac{1}{q^{n}}
        \sum_{i = 0}^{d - 1} 
        i \binom{n - i}{\tau_{i}} 
        (q - 1)^{i + \tau_{i}}, 
    \end{align*}
    where $\tau_{i}$ is the decoding radius of $\GRSdecode(\CC_{i}, y^{(i)})$. 
\end{theorem}

\begin{proof}
    Letting $f \in \CF(k - 1, q)$ with $\dH(y, \CC(f))$ minimal, 
    \begin{align*}
        \Pr[P(\CC, y) = i] 
        &= \Pr[y_{j} \neq f(\alpha_{j}) \hbox{ for $\tau_{i}$ values of } j \le n - i \\ 
        & \hspace{1.5cm}\hbox{and }  y_{j} \neq f(\alpha_{j})
        \hbox{ for all $j > n - i$}
        ] \\ 
        &= 
        \binom{n - i}{\tau_{i}} 
        \left(1 - \frac{1}{q}\right)^{\tau_{i} + i} 
        \left(\frac{1}{q}\right)^{n - i - \tau_{i}}. 
    \end{align*}
    \hide{
    So the expected number of punctures needed for \autoref{alg:GRS-cover} to succeed is 
    \begin{align*}
        \sum_{i = 0}^{d - 1} 
        i \binom{n - i}{\tau} 
        \left(1 - \frac{1}{q}\right)^{\tau + i} 
        \left(\frac{1}{q}\right)^{n - i - \tau}. 
    \end{align*}
    }(Alternatively, 
    the probability that a uniformly random $y \in \fq^{n}$ satisfies $y_{j} \neq f(\alpha_{j})$ for $\tau_{i}$ values of $j \le n - i$ and $y_{j} \neq f(\alpha_{j})$ for all $j > n - i$ is 
    \begin{align*}
        \frac{1}{q^{n}}
        \binom{n - i}{\tau_{i}} 
        (q - 1)^{\tau_{i} + i}, 
    \end{align*}
    which is equal to $\Pr[P(\CC, y) = i]$ above.) 
    \hide{So the expected number of steps for $\GRScover(\CC, y)$ to succeed is 
    \begin{align*}
        \frac{1}{q^{n}}
        \sum_{i = 0}^{d - 1} i 
        \binom{n - i}{\tau} 
        (q - 1)^{\tau + i}. 
    \end{align*}
    }The conclusion now follows from the definition of the expected value. 
\end{proof}
Assuming that $y$ is chosen uniformly at random, 
\begin{align*}
    &\Pr[P(\CC, y) = i] 
    = \Pr[\GRSdecode(\CC_{i}, y^{(i)})\ {\rm succeeds}] \\ 
    &\qquad\quad \cdot \prod_{j = 0}^{i - 1} \Pr[\GRSdecode(\CC_{j}, y^{(j)})\ {\rm fails}] \\
    &= \frac{|\bigcup_{c \in \CC_{i}} B(c, \tau_{i})|}{q^{n - i}}
    \prod_{j = 0}^{i - 1} 
    \left(1 - \frac{|\bigcup_{c \in \CC_{j}} B(c, \tau_{j})|}{q^{n - j}}\right) \\ 
    &
    \resizebox{\hsize}{!}{
    $\displaystyle \le  
    \frac{\Vol_{q}(\tau_{i}, n - i)}{q^{n - k - i}} 
    \prod_{j = 0}^{i - 1} 
    \left(1 - \frac{ 
    \Vol_{q}(\tau_{j}, n - j) 
    - \frac{1}{2} \sum_{w = d_{j}}^{2 \tau_{j}} A_{w} I(w, \tau_{j})}{q^{n - k - j}}\right).$}
\end{align*}
where $\tau_{i}$ is the decoding radius of $\GRSdecode(\CC_{i}, y^{(i)})$. 
} 


\subsection{Complexity Analysis}

Plainly, the worst-case running time of \autoref{alg:GRS-cover} is $(n - k) \cdot T(n)$, where $T(n)$ is the running time of $\GRSdecode$. 

\begin{itemize}
    \item For unique decoding, algorithms like Berlekamp--Welch (BW) are known with $T(n) = O(n^{3})$~\cite[Theorem~15.1.4]{ecc} as well as more efficient ones with $T(n) = O(n^{2})$ and $T(n) = O(n \polylog(n))$~\cite[Section~15.4]{ecc}. 
    \item For list decoding, there are known implementations of $\GS$ with $T(n) = O(s^{4} n^{2})$~\cite[\S~VII]{McEliece03}, \cite[\S~VII]{McEliece03}, \cite{Shokrollahi00}, \cite[\S~V]{Roth00},\hide{solving the interpolation problem 
    in GS takes $O(n^{3})$ time using the Feng--Tzeng algorithm \cite[\S~VIII]{McEliece03} or naive Gaussian elimination. This was improved by K\"{o}tter to $O(s^{4} n^{2})$ \cite[\S~VII]{McEliece03}, which is the most efficient known solution~
    \cite[\S~4]{McEliece03-1}. 
    
    The most efficient known solution to the factorization problem 
    of GS is due to Gao and Shokrollahi \cite{Shokrollahi00} with a time complexity of $O(\ell^{3} k^{2})$, 
    although the Roth--Ruckenstein algorithm \cite[\S~V]{Roth00} is quite competetive \cite[\S~4]{McEliece03-1} 
    (see also \cite[\S~IX]{McEliece03}). 
    Here $\ell$ is a design parameter, typically a small constant \cite{Roth00}, which is an upper bound on the size of the list of decoded codewords, which is bounded above by the degree of the interpolation polynomial in the second variable.
    ~\cite{Guruswami06}. 
    
    In general, $\ell \le \lfloor c / (k - 1) \rfloor
    = O(s \sqrt{n / k})$ (see also \cite[Theorem~4.8]{Guruswami06}), which gives $O(\ell^{3} k^{2}) = O(s^{3} n \sqrt{k n})$, yielding a $O(s^{4} n^{2})$ time complexity for \autoref{alg:GRS-cover}.}
    where $s > 0$ is an integer parameter known as the \emph{interpolation multiplicity}\hide{, $\GS(y, s)$ 
    returns all $f \in \CF(k - 1, q)$ such that $\dH(y, \GRS(f)) \le n - \tau_{s}$, 
    where $\tau_{s} := \lfloor c / s \rfloor + 1$ and $c := \lfloor \sqrt{(k - 1) n s (s + 1)} \rfloor$}. 
\end{itemize}

However, the average-case running time of \autoref{alg:GRS-cover} depends on the average number of punctures needed for $\GRSdecode$ to succeed, 
and based on the data in \autoref{tab:punctures}, we make the following conjecture. 


\begin{conjecture}
    \label{conj:punctures}
    For an $[n, k, d]_{q}$ GRS code, the average number of punctures needed for \autoref{alg:GRS-cover} to succeed in returning a codeword within the covering radius is less than 
    \begin{itemize}
        \item $c_{1} (d - 1)$ for $\GRSdecode$ a unique decoder, and 
        \item $c_{2}$ for $\GRSdecode$ the $\GS$ list decoder, 
    \end{itemize}
    for some positive constants $c_{1}, c_2 <1$. 
\end{conjecture}

Note that \autoref{conj:punctures} would imply that the average covering radius using \autoref{alg:GRS-cover} is a constant fraction (less than one) of the covering radius, which can happen in the worst case, when a unique decoder is deployed. Furthermore, when the GS list decoder is deployed, then the average covering radius is within a constant from $\tau_{\GS}$. 



\hide{
\begin{algorithm}[!htbp]
    \caption{$\GRScover(\CC, y)$} 
    \label{alg:GRS-cover-2}
    \begin{algorithmic}[1]
        \Require{$[n, k, d]_{q}$ $\GRS$ code $\CC$, arbitrary vector $y \in \fq^{n}$}
        \Ensure{A codeword $c \in \CC$ with $\dH(y, c) \le d - 1$}
        \BlackBox{$\GRSdecode$: a $\GRS$ list decoder}
        \State{$\CC_{0} \gets \CC$}
        \For{$i = 1, \dots, d-1$}
            \State $L \gets \GRSdecode(\CC, y)$
            \For {$f \in L$}
                \If {$\dH(y, \CC_{0}(f)) \le d - 1$} \label{alg:2:5}
                    \State{\Return $\CC_{0}(f)$} 
                \EndIf
            \EndFor
            \State{$\CC \gets \hbox{puncturing of $\CC$ at coordinate $n - i + 1$}$}
            \State{$y \gets y[1..n - i]$} 
        \EndFor
    \end{algorithmic}
\end{algorithm}

\question{Improving GS interpolation step for $\CCP$?}

\answer{Let $f \in \fpkq$. Then 
\begin{align*}
    0 
    &\equiv Q(X, f(X)) \\  
    &= \sum_{i, j} q_{i, j} X^{i} \left(\sum_{\substack{1 \le \ell \le k \\ p \nmid \ell}} f_{\ell} X^{\ell}\right)^{j} \\ 
    &= \sum_{i, j} q_{i, j} 
    \sum_{\substack{1 \le \ell_{1}, \dots, \ell_{j} \le k \\ 
    p \nmid \ell_{1} \cdots \ell_{j}}} 
    f_{\ell_{1}} \cdots f_{\ell_{j}} X^{i + \ell_{1} + \cdots + \ell_{j}} \\ 
    &= \sum_{t = 0}^{\ell} 
    Q^{(t)}(X) f(X)^{t}. 
\end{align*}
Then $Q^{(0)}(X) = 0$. Since each coefficient must vanish individually, we have, for each $d \ge 0$, 
\begin{align*}
    \sum_{i, j} q_{i, j} 
    \sum_{\substack{1 \le \ell_{1}, \dots, \ell_{j} \le k \\ 
    p \nmid \ell_{1} \cdots \ell_{j} \\ 
    i + \ell_{1} + \cdots + \ell_{j} = d}} 
    f_{\ell_{1}} \cdots f_{\ell_{j}} 
    &= 0.
\end{align*}
}
}

\subsection{Coverage Fraction} 

\hide{
Assuming that points in $\fqn$ are independently covered by each $B(c, \tau)$, the probability that a randomly chosen point is not covered by any $B(c, \tau)$ is 
\begin{align*}
    \prod_{c \in \CC} 
    \left(1 - \frac{|B(c, \tau)|}{|\fqn|}\right) 
    &= \left(1 - \frac{\sum_{j = 0}^{\tau} \binom{n}{j} 
    (q - 1)^{j}}{q^{n}}\right)^{q^{k}}. 
\end{align*}
Hence, the number of points in $\fqn$ covered by the $B(c, \tau)$ is \hide{
\begin{align*}
    \bigabs{\bigcup_{c \in \CC} B(c, \tau)}. 
\end{align*}
Plainly\footnote{See \href{https://feog.github.io/11-coding.pdf}{here}.}
\begin{align*}
    \bigabs{\bigcup_{c \in \CC} B(c, \tau)} 
    &\le |B(c, \tau)| q^{k} 
    = q^{k} \sum_{j = 0}^{\tau} \binom{n}{j} (q - 1)^{j}
\end{align*}
and\footnote{See Section 3.2 \href{https://www.researchgate.net/publication/360859226_On_the_number_of_error_correcting_codes/fulltext/628f10fb55273755ebb5b1a0/On-the-number-of-error-correcting-codes.pdf?origin=scientificContributions}{here}.} 
\begin{align*}
    \bigabs{\bigcup_{c \in \CC} B(c, \tau)} 
    &\ge |B(c, \tau)| q^{k} - \binom{q^{k}}{2} 
    W_{q}(n, k, \tau). 
\end{align*}
In particular,} 
\begin{align*}
    \bigabs{\bigcup_{c \in \CC} B(c, \tau)} 
    &\approx q^{n} 
    \left(1 - \left(1 - \frac{\sum_{j = 0}^{\tau} \binom{n}{j} 
    (q - 1)^{j}}{q^{n}}\right)^{q^{k}}\right) \\ 
    &\approx q^{k} 
    \sum_{j = 0}^{\tau} \binom{n}{j} 
    (q - 1)^{j} \\ 
    &\approx q^{k + n h_{q}(\tau / n)} 
\end{align*}
for $\tau / n \le 1 - 1 / q$. 
When $\tau = \tau_{\GS}$, this is asymptotically equal to $q^{n (R + h_{q}(1 - \sqrt{R}))} \approx q^{n}$, where $R = k / n$. 
Hence, \hide{for $R$ sufficiently large, }$\GRScover$ will almost always succeed without any punctures.

\question{Average covering radius?}

\answer{Given a block code $\CC \subseteq \fqn$, it follows from \eqref{eq:rho} that every $y \in \fqn$ satisfies $d(y, c) \le \rho(\CC)$ for some $c \in \CC$. Moreover, $d_{\min}(\CC_{i}) = d - i$ by \eqref{eq:dmin-reduction}. Hence, by \autoref{alg:2:5} of \autoref{alg:GRS-cover-2}, ... 

In general, we need to solve 
\begin{align*}
    k + n h_{q}(\tau / n) 
    &\ge n 
\end{align*}
for $\tau$. 
}

\question{Improvement using weight distribution of RS?}

}
Next, we show a lower bound on the size/fraction of the space covered by Hamming spheres of a given radius centered at the codewords of an MDS code in terms of its weight distribution. The key observation is that one can explicitly compute the sizes of the Hamming spheres and their pairwise intersections, as \autoref{thm:union-lower-bound} below demonstrates. 


\begin{theorem}
    \label{thm:union-lower-bound}
    Let $\CC$ be an $[n, k, d]_{q}$ MDS code. Then
    \begin{align*}
        \bigabs{\bigcup_{c \in \CC} B(c, \tau)} 
        \ge q^{k} \left(\sum_{i = 0}^{\tau} \binom{n}{i} (q - 1)^{i} 
        - \frac{1}{2} \sum_{w = d}^{2 \tau} A_{w} I(w, \tau)\right), 
    \end{align*}
    where 
    \begin{align}
        &A_{w} 
        = \binom{n}{w} 
        \sum_{j = 0}^{w - d} (-1)^{j} 
        \binom{w}{j} (q^{w - d + 1 - j} - 1)  \label{eq:weight-distribution} 
    \end{align} 
    and 
    \begin{align}
        & I(w, \tau) = 
        \label{eq:intersection} \\ 
        &\hspace{-1mm} 
        \resizebox{\hsize}{!}{
        $\displaystyle \sum_{z = 0}^{n - w} 
        \binom{n - w}{z} 
        (q - 1)^{n - w - z}
        \hspace{-3mm}
        \sum_{
            \substack{
                n - \tau - z \le u, v \le \tau  
                \\ 
                u + v \le w 
            }
        } \binom{w}{u} 
        \binom{w - u}{v}
        (q - 2)^{w - u - v}$.} 
        \notag 
    \end{align} 
\end{theorem}

\begin{proof}
    By the inclusion-exclusion principle, 
    \begin{align*}
        \bigabs{\bigcup_{c \in \CC} B(c, \tau)} 
        &\ge \sum_{c \in \CC} |B(c, \tau)| - \frac{1}{2} \sum_{c_{1} \neq c_{2}} |B(c_{1}, \tau) \cap B(c_{2}, \tau)| \\ 
        &= |\CC| \left(\sum_{i = 0}^{\tau} \binom{n}{i} (q - 1)^{i} 
        - \frac{1}{2} \sum_{w = d}^{2 \tau} A_{w} I(w, \tau)\right), 
    \end{align*}
    where $A_{w}$ is the number of codewords $c \in \CC$ of weight $w$, and $I(w, \tau) := |B(c_{1}, \tau) \cap B(c_{2}, \tau)|$ for $c_{1}, c_{2} \in \CC$ with $\dH(c_{1}, c_{2}) = w$. 
    Now, \eqref{eq:weight-distribution} is known to hold for $d \le w \le n$ (e.g., see~\cite{Macwilliams77}), and \eqref{eq:intersection} follows by a counting argument, where $u$ (resp. $v$) represents the number of indices in $\{1, \dots, n\}$ where $y$ and $c_{1}$ (resp. $c_{2}$) agree, and $z$ represents the number of indices where $c_{1}$, $c_{2}$ and $y$ agree. 
\end{proof}



\autoref{thm:union-lower-bound} gives the following bounds on the fraction of space covered by the union of the Hamming spheres. 


\begin{corollary}
    \label{cor:union-lower-bound} 
    Let $\CC$ be an $[n, k, d]_{q}$ GRS code. Then
    the fraction of the ambient space covered by the union of Hamming spheres of radius $\tau$ centered at the codewords satisfies 
    \begin{align*}
        & \resizebox{\hsize}{!}{$\displaystyle q^{k - n} 
        \Vol_{q}(\tau, n) 
        - \frac{q^{k - d}}{2}  
        \sum_{w = d}^{2 \tau} 
        \binom{n}{w} 
        [\Vol_{q}(\tau, w) - \Vol_{q}(w - \tau - 1, w)]$} 
        \\ 
        &\le 
        q^{-n} \bigabs{\bigcup_{c \in \CC} B(c, \tau)} 
        \le q^{k - n} 
        \Vol_{q}(\tau, n),  
    \end{align*}
    where (see, e.g., \cite{ecc})
    \begin{align*}
        \Vol_{q}(\tau, n) 
        &:= |B(0, \tau)| 
        = \sum_{j = 0}^{\tau} 
        \binom{n}{r} (q - 1)^{j} 
        \approx q^{n H_{q}(\tau / n)}
    \end{align*}
    for $H_{q}(x) := x \log_{q}(q - 1) - x \log_{q}(x) - (1 - x) \log_{q}(1 - x)$. 
\end{corollary}

\begin{proof}
    Observe that 
    \begin{align*}
        A_{w} 
        &\le \binom{n}{w} q^{w - d} 
    \end{align*}
    by \eqref{eq:weight-distribution}, 
    and 
    \begin{align*}
        & I(w, \tau) 
        \ \\ 
        &\resizebox{\hsize}{!}{$\displaystyle = \sum_{z = 0}^{n - w} 
        \binom{n - w}{z} 
        (q - 1)^{n - w - z}
        \sum_{
                n - \tau - z \le u \le \tau 
        } \binom{w}{u} 
        \sum_{
                n - \tau - z \le v \le \tau 
        } 
        \binom{w - u}{v}
        (q - 2)^{w - u - v}$} \\ 
        & 
        \le \sum_{z = 0}^{n - w} 
        \binom{n - w}{z} 
        (q - 1)^{n - w - z}
        \sum_{
                n - \tau - z \le u \le \tau 
        } \binom{w}{u} 
        (q - 1)^{w - u}
        \\ 
        &\resizebox{\hsize}{!}{$\displaystyle = \sum_{z = 0}^{n - w} 
        \binom{n - w}{z} 
        (q - 1)^{n - w - z} 
        [\Vol_{q}(\tau, w) - \Vol_{q}(n - \tau - z - 1, w)]$} \\ 
        &\le 
        q^{n - w} 
        [\Vol_{q}(\tau, w) - \Vol_{q}(w - \tau - 1, w)] 
    \end{align*}
    by \eqref{eq:intersection}. The conclusion now follows from \autoref{thm:union-lower-bound}. 
\end{proof}

Note that the bound in \autoref{thm:union-lower-bound} is sharp when overlaps of three or more Hamming spheres are negligible, which is the case when $\tau$ is sufficiently small, and, in particular, expected to be the case for the smallest $\tau$ such that the union of Hamming spheres of radius $\tau$ covers most of the space. We estimated this bound for $\tau \in (d / 2, d)$ for several $[n, k, d]_{q}$ GRS codes and observed that the radius $\tau_{\max}$ yielding the best lower bound is generally close 
to $\tau_{\GS}$. 
In the next section, we will show numerically that 
\begin{itemize}
    \item $\tau_{\GS} \approx \tau_{\max}$ for all $1 \le k < n$ with $(q, n) = (47, 46)$, 
    implying that Hamming spheres of radius $\tau_{\GS}$ cover a significant fraction of the ambient space,
    \item \autoref{conj:punctures} is satisfied for all $1 \le k < n$ with $(q, n) = (7, 6)$, $(11, 10)$, 
    and
    \item the average covering radius obtained via $\GRScover$ using the $\GS$ list decoder is very close to that obtained using a \emph{maximum-a-posteriori} (MAP) decoder for $(q, n) = (7, 6)$ and $1 \le k < n$.
\end{itemize}
\hide{
\todo{try to show that for $d / 2$ it's close to $0$ and for $d$ it is $1$; ignoring overlaps, when does it reach $q^{n - k}$?}

Assuming that the overlaps between the Hamming spheres are small, this is roughly equal to $q^{k - n + n H_{q}(\tau / n)}$. This is equal to $1$ when $H_{q}(\tau / n) = 1 - k / n$, i.e. 
\begin{align*}
    \tau 
    &\approx (q - 1) (n - (\log q)^{2} (d + 1)) 
\end{align*}
and for $\tau = d / 2$, 
\begin{align*}
    k - n + n H_{q}(d / (2 n)) 
    \approx k - n + d / 2 
    = 1 - d / 2, 
\end{align*}
which tends to $-\infty$ for sufficiently large $n$ and $k \ll n$. 
Here, we have used the approximation $x \approx (q - 1) q^{- \log q (1 - k / n)}$ when $1 - k / n$ is small.} 


\hide{In particular,\footnote{See also \href{https://dspace.mit.edu/bitstream/handle/1721.1/36834/6-451Spring-2003/NR/rdonlyres/Electrical-Engineering-and-Computer-Science/6-451Spring-2003/84820E55-5147-40DC-83C4-1A090FC63D56/0/chapter8.pdf}{here}.} 
\begin{align*}
    A_{d} 
    &= \binom{n}{d} (q - 1), \\ 
    A_{d + 1} 
    &= \binom{n}{d + 1} 
    \left(q^{2} - 1 - \binom{d + 1}{d} (q - 1)\right) \\ 
    &= \binom{n}{d + 1} 
    (q - 1) (q - d)
    . 
\end{align*}
\todo{It seems that the maximum of the RHS of \eqref{eq:union-lower-bound} occurs when $\tau = n - k - g(n) \gtrsim \tau_{\GS}$.}

Observe that\footnote{See \href{https://math.stackexchange.com/questions/887960/truncated-alternating-binomial-sum}{here}.} 
\begin{align*}
    &A_{w} 
    = \binom{n}{w} 
    \bigg(q^{w - d + 1} \sum_{j = 0}^{w - d} \binom{w}{j} (-q)^{-j} - \sum_{j = 0}^{w - d} (-1)^{j} \binom{w}{j}\bigg) 
    \\ 
    &= \binom{n}{w} 
    \sum_{j = 0}^{w - d} \binom{w}{j} (-1)^{j} (q^{w - d + 1 - j} - 1) 
    \\ 
    &= \binom{n}{w} 
    \bigg((-1)^{w - d} \binom{w}{d - 1} \,_{2}F_{1}(1, 1 - d; w - d + 2; 1 / q) \\ 
    &\qquad \qquad + \frac{(q - 1)^{w}}{q^{d - 1}} 
    - (-1)^{w - d} \binom{w - 1}{d - 1}\bigg) \\ 
    &= \binom{n}{w} 
    \bigg(\frac{(q - 1)^{w}}{q^{d - 1}} + (-1)^{w - d} \binom{w - 1}{d - 1} \\ 
    & \qquad \qquad  \bigg(\frac{\,_{2}F_{1}(1, 1 - d; w - d + 2; 1 / q) w}{w - d + 1} - 1\bigg)\bigg)
    .
\end{align*}
Now, we want a $\tau$ such that 
\begin{align*}
    |B(c, \tau)| 
    = \Vol_{q}(\tau, n) 
    &\le q^{d}, \\ 
    \sum_{w = d}^{2 \tau} 
    A_{w} f(\tau, w) 
    &\ll q^{d}. 
\end{align*}
Note that \hide{By inspection, $\,_{2}F_{1}(1, 1 - d; w - d + 2; 1 / q) \lesssim 1$. Hence, 
\begin{align*} 
    A_{w} 
    &\le 
    \binom{n}{w} q^{w - d + 1} 
\end{align*}
and }
\begin{align*}
    f(\tau, w) 
    &\le q^{n - w + w H_{q}(\tau / w)}
\end{align*}
so that 
\begin{align*}
    &\sum_{w = d}^{2 \tau} A_{w} f(\tau, w) \\ 
    &\le \sum_{w = d}^{2 \tau} 
    q^{n - w + w H_{q}(\tau / w)} 
    \sum_{j = 0}^{w - d} 
    \binom{w}{j} (-1)^{j} (q^{w - d + 1 - j} - 1) 
    \\ 
    &\le \sum_{w = d}^{2 \tau} 
    \binom{n}{w}
    q^{k + w H_{q}(\tau / w)} \\ 
    &\approx 
    q^{k} \sum_{w = d}^{2 \tau} 
    \binom{n}{w} 
    \sum_{j = 0}^{\tau} 
    \binom{w}{j} 
    (q - 1)^{j}. 
\end{align*}
On the other hand, $|B(c, \tau)| = O(q^{n H_{q}(\tau / n)})$.
So, we want $\tau$ such that $n H_{q}(\tau / n) \le d$ and $3 \tau - d + 1 < d$, i.e. $H_{q}(\tau / n) \le d / n$ and $\tau < (2 d - 1) / 3$.
Therefore, it suffices to have 
\begin{align*}
    &n H_{q}(\tau / n) 
    \le d 
    . 
\end{align*}

\subsection{Theoretical Results}

\question{
(Idea) Take a random $y \in \fq^{n}$. What is the probability that $y$ is within distance $r$ of some $\GRS$ codeword? 
}

\answer{
Let this probability be $P_{r}$ and denote $[n] := \{1, \dots, n\}$. 
\hide{
Note that 
\begin{align*}
    \Pr\left[\bigwedge_{i \in I} \{y_{i} = f(\alpha_{i})\}\right] 
    &= \prod_{i \in I} \Pr[y_{i} = f(\alpha_{i})]. 
\end{align*}
}
Define 
\begin{align*}
    g_{f}(Y) 
    &= 
    \prod_{i > n - r} 
    (Y_{i} - f(\alpha_{i})). 
\end{align*}
Note that $g_{f}(Y) = 0$ if $Y_{i} = f(\alpha_{i})$ for some $i > n - r$. 
By the Schwartz--Zippel lemma, 
\begin{align*}
    P_{r} 
    \le \sum_{f \in \CF(k - 1, q)} \Pr[g_{f}(Y) = 0] 
    &\le \frac{\binom{n}{n - r}}{q^{n - k}}. 
\end{align*}
Now let $r = \tau_{\GS}$. 
}
}

\subsection{Quantization Using CP Codes}
\label{sec:quantization}

As mentioned in the introduction, CP codes have recently been used for the quantization problem in MISO systems with limited feedback~\cite{gooty2025precodingdesignlimitedfeedbackmiso}, 
and the quantization error was characterized using the covering radius of CP codes. 
Moreover, a decoding algorithm for CP codes over prime fields was recently proposed~\cite{riasat2024decodinganalogsubspacecodes}. 
However, as discussed in \autoref{sec:packing-covering} and at the beginning of \autoref{sec:covering}, decoding algorithms are not sufficient for the purpose of covering. Therefore, in future work, we wish to extend the above results to CP codes. 

\section{Numerical Results}
\label{sec:simulation}

\subsection{$\tau_{\GS}$ Versus $\tau_{\max}$}

\autoref{fig:coverage} shows that $\tau_{\GS} = \tau_{\max}$ for $(q, n, k) = (17, 14, 2)$, and that a significant fraction of the space is covered by Hamming spheres of radius close to $\tau_{\GS}$. 
\autoref{fig:comparison} shows a comparison of $\tau_{\max}$ with $\tau_{\GS}$ for $(q, n) = (47, 46)$ and $1 \le k < n$, 

\begin{figure}[!htbp]
    \centering
    \hide{
    \begin{subfigure}[t]{.38\textwidth}
        \includegraphics[width=\linewidth]
        {coverage_16_2.png}
        \caption{$q = 16, n = 14, k = 2$}
    \end{subfigure}
    ~
    \begin{subfigure}[t]{.38\textwidth}
        \includegraphics[width=\linewidth]
        {coverage_16_3.png}
        \caption{$q = 16, n = 14, k = 3$}
    \end{subfigure}
    }
    \includegraphics[width=.9\linewidth]{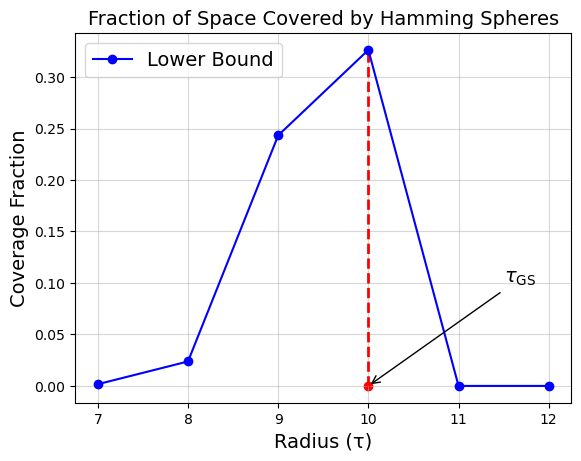}
    \caption{Lower bound on fraction of space covered by Hamming spheres of radius $\tau \in (d / 2, d)$ for $[14, 2, d]_{17}$ GRS code
    }
    \label{fig:coverage}
\end{figure}

\begin{figure}[!htbp]
    \centering
    \includegraphics[width=.95\linewidth]{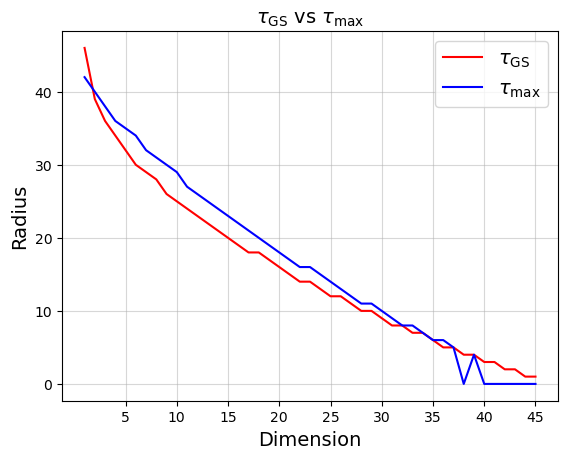}
    \caption{Comparison of $\tau_{\max}$ with $\tau_{\GS}$ for $q = 47, n = 46$ and $1 \le k < n$}
    \label{fig:comparison}
\end{figure}

\subsection{Average Number of Punctures}

\autoref{tab:punctures} lists the empirical average number of punctures for $\GRScover$ to succeed with $\GRSdecode$ a unique decoder (BW) and a list decoder (GS), 
for $(q, n) = (7, 6)$, $(11, 10)$ and $1 \le k < n$. For each $k$, the average was computed over $500$ simulations. 
It shows that \autoref{alg:GRS-cover} succeeds using BW usually within a few punctures, but most of the time without any punctures using GS, in accordance with \autoref{conj:punctures}. 
In other words, while Hamming spheres of radius $< d / 2$ do not cover the space well, those of radius $\tau_{\GS}$ cover a large fraction of the space. 

\begin{table}[!htbp]
    \caption{Average number of punctures before $\GRScover$ succeeds using unique (BW) and list (GS) decoders} 
    \begin{subtable}[t]{0.5\linewidth}
        \centering
        \begin{tabular}[t]{ccc}
            \toprule 
            $k$ & BW & GS \\
            \midrule 
            1 & 3.836 & 0 \\
            2 & 2.44 & 0.006 \\
            3 & 1.654 & 0.08 \\
            4 & 0.532 & 0.264 \\
            5 & 0.868 & 0 \\
            
            \bottomrule
    \end{tabular}
    \caption{$q = 7$, $n = 6$}
    \end{subtable}%
    \begin{subtable}[t]{0.5\linewidth}
        \centering
        \begin{tabular}[t]{ccc}
            \toprule 
            $k$ & BW & GS \\
            \midrule 
            1 & 8.224 & 0 \\
            2 & 6.872 & 0.452 \\
            3 & 5.654 & 0.366 \\
            4 & 4.34 & 0.378 \\
            5 & 3.022 & 0.7 \\
            6 & 1.864 & 0.744 \\ 
            7 & 1.428 & 0.036 \\
            8 & 0.376 & 0.16 \\
            9 & 0.922 & 0 \\
            \bottomrule
        \end{tabular}
        \caption{$q = 11$, $n = 10$}
        \end{subtable}
    \label{tab:punctures}
\end{table}

\subsection{Average Covering Radius}
\label{sec:average-covering-radius}

To estimate the average covering radius $\bar\rho(\CC)$ defined in \eqref{eq:rhobar}, we computed the average value of $\min_{c \in \CC} \dH(y, c)$\hide{covering radius $\bar\rho(\CC)$} over $500$ simulations with $\GRSdecode$ being a MAP decoder, a unique decoder (BW), and a list decoder (GS). The results are compared in \autoref{fig:covering-radius} for $q = 7, n = 6$ and $1 \le k < n$ \hide{and $q = 11, n = 10$, $1 \le k < n - 2$ }with \autoref{alg:trivial-cover} as baseline, as well as the actual covering radius $\rho(\CC)\hide{ = d - 1}$. 
It shows that the performance of \autoref{alg:GRS-cover} using BW already surpasses that of the baseline \hide{the trivial algorithm}\autoref{alg:trivial-cover}. This is then further improved when the GS list decoding is deployed. Furthermore, the performance of \autoref{alg:GRS-cover} with the GS list decoder almost matches that of the MAP decoder. Note that the MAP decoder provides the optimal covering solution with an exponential complexity.



\begin{figure}[!htbp]
    \hide{
    \begin{subfigure}{.5\textwidth}
    }
    \centering
    \includegraphics[width=\linewidth]{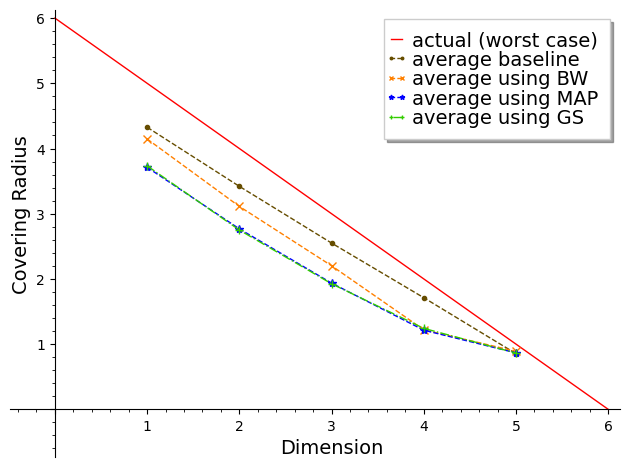}
    \hide{
    \caption{$q = 7$, $n = 6$}
    \end{subfigure}
    ~
    \begin{subfigure}{.5\textwidth}
    \centering
    \includegraphics[width=\textwidth]{11-10-map.png}
    \caption{$q = 11$, $n = 10$}
    \end{subfigure}
    }
    \caption{Comparison of simulated average covering radius of a $[6, k]_{7}$ GRS code for $1 \le k < 6$ using different algorithms}
    \label{fig:covering-radius}
\end{figure}

\section{Conclusion and Future Directions}
\label{sec:conclusion}

In this work, we presented a new algorithm to efficiently cover the ambient space using GRS codes. 
\autoref{conj:punctures} naturally stems from the algorithm and is strongly supported by numerical results, and it would be ideal to be able to prove it analytically. 
As discussed in \autoref{sec:quantization}, our goal 
is to extend the results of this paper to CP codes, and more generally to Grassmann codes. 
Since CP codes are constructed from particular subcodes of GRS codes, their covering radii are directly connected. 
However, adapting our algorithm to CP codes will likely result in an increase in the number of punctures, and we wish to explore this in more detail in a future work. 


\bibliographystyle{IEEEtran}
\bibliography{IEEEabrv, references}

\end{document}